\pdfoutput=1
\documentclass[runningheads]{llncs}
\usepackage[hidelinks]{hyperref}
\hypersetup{
  pdftitle={Stochastic Languages at Sub-stochastic Cost},
  pdfauthor={Smayan Agarwal, Aalok Thakkar},
  pdfkeywords={quantitative automata, cost register automata, stochastic languages, spectral methods},
  pdfsubject={Preprint version},
  pdfcreator={LaTeX on Overleaf},
  pdfproducer={Overleaf, pdfTeX},
  colorlinks=true,
  linkcolor=black,
  citecolor=black,
  urlcolor=black
}

\usepackage[T1]{fontenc}
\usepackage{amsmath}
\usepackage{amssymb}
\usepackage{stmaryrd}
\usepackage[inline]{enumitem}
\usepackage{tikz}
\usetikzlibrary{automata, positioning, shapes.multipart,arrows.meta}
\usepackage{float}
\usepackage{graphicx}
\newcommand{\R}{\mathbb R}
\newcommand{\Rp}{\mathbb{R}_{\geq 0}}

\newcommand{\semantics}[1]{{\llbracket #1 \rrbracket}}
\newcommand{\calA}{\mathcal{A}}

\newcommand{\mass}[2]{\mathrm{mass}\left(#2\right)_{#1}}

\newcommand{\Stoch}{\mathcal{S}}
\newcommand{\Srat}{\mathcal{S}_{\Rp}^{rat}(\Sigma^*)}

\begin{document}

\title{Stochastic Languages at Sub-stochastic Cost}

\author{Smayan Agarwal, Aalok Thakkar}

\authorrunning{Agarwal \& Thakkar}
\institute{Ashoka University}
\maketitle

\begin{abstract}
When does a deterministic computational model define a probability distribution? What are its properties? This work formalises and settles this stochasticity problem for weighted automata, and its generalisation cost register automata (CRA).

We show that checking stochasticity is undecidable for CRAs in general. This motivates the study of the fully linear fragment, where a complete and tractable theory is established. For this class, stochasticity becomes decidable in polynomial time via spectral methods, and every stochastic linear CRA admits an equivalent model with locally sub-stochastic update functions. This provides a local syntactic characterisation of the semantics of the quantitative model.

This local characterisation allows us to provide an algebraic Kleene-Schützenberger characterisation for stochastic languages. The class of rational stochastic languages is the smallest class containing finite support distributions, which is closed under convex combination, Cauchy product, and discounted Kleene star. We also introduce Stochastic Regular Expressions as a complete and composable grammar for this class.

Our framework provides the foundations for a formal theory of probabilistic computation, with immediate consequences for approximation, sampling, and distribution testing.
\end{abstract}

\section{Introduction}
\label{sec:intro}

A fundamental question in formal language theory asks: \emph{Which functions on strings give rise to probability distributions?} More precisely, given a quantitative automaton that assigns real-valued weights to words through algebraic computations, when does this function, after normalization, define a valid probability measure over $\Sigma^*$? This stochasticity problem sits at the confluence of automata theory, linear algebra, and probability theory, connecting classical questions about regularity with modern challenges in probabilistic modelling.

The question is deceptively simple. Unlike classical probabilistic automata~\cite{rabin1963probabilistic,paz1971probabilistic}, which embed randomness directly into their operational semantics, weighted automata~\cite{droste2007weighted} and cost register automata (CRAs)~\cite{Alur2013} compute weights by accumulating numeric values through matrix multiplication or register updates, producing a function $f:\Sigma^*\to \mathbb{R}$ without internal probabilistic choices. The stochasticity question asks whether this deterministically computed function satisfies two global conditions: non-negativity ($f(w) \geq 0$ for all $w$) and normalization ($\sum_{w \in \Sigma^*} f(w) = 1$). Answering this requires reasoning about infinite summations over all possible strings--a task that, as we show, is undecidable in general.

\paragraph{Motivation.} 
Probability distributions over strings pervade computer science. In natural language processing, stochastic grammars model syntactic structure~\cite{chomsky1956,jelinek1998}, while neural language models implicitly define distributions over token sequences~\cite{brown1993}. Biological sequence analysis relies on probabilistic models of DNA and protein strings, where hidden Markov models capture evolutionary dynamics~\cite{durbin1998}. Formal verification employs probabilistic specifications to reason about quantitative properties~\cite{baier2008,uppaal2004}. Information theory studies optimal compression via string distributions~\cite{cover2006}.
Quantitative automata theory offers formal models for their study. When they define stochastic languages, they provide interpretable, verifiable probability models with decidable properties and compositional structure. Understanding \emph{when} they are stochastic, and \emph{why}, is thus both theoretically natural and practically consequential.

\paragraph{Results and Contributions.}
Our focus on deterministically computed distributions provides a unifying framework for these domains. We study cost register automata as a model of computing probability distributions over strings. Our main contributions are:

\begin{enumerate}
    \item We prove that determining whether a CRA defines a stochastic language is undecidable in general (Theorems~\ref{thm:cra-undecidability} and~\ref{thm:poly-cra-undecidability}). To recover decidability, we identify and analyse the fragment of CRAs with affine updates. We establish that stochasticity decision problem decidable for this fragment (Theorem~\ref{thm:affine-decidable}). 
    \item Our second main contribution proves that every linear CRA defining a stochastic language admits a semantically equivalent representation in which stochasticity becomes local and can be checked by computing the sum of transition weights of each register (Theorem~\ref{thm:local-substochastic}). 
    \item We provide a Kleene-Schützenberger characterization for rational stochastic languages (Theorem~\ref{thm:kleene-characterisation}), and introduce stochastic regular expressions as a grammar for rational stochastic languages.
\end{enumerate}

Our results chart a precise boundary between decidable and undecidable semantic properties of quantitative automata and provide sub-stochastic and local characterisation for CRA with linear updates. 
We also present three open questions on 
\emph{stochasticity decision problem} for general quantitative models, \emph{minimal approximation methods}, and \emph{distribution testing} for stochastic languages. 
These results and problems bridge automata theory and statistical inference, 
opening new directions for both theoretical investigation, modelling of stochastic sequences, and validating probabilistic generative models against formal distributional specifications.

\section{Related Work}
\label{sec:related-work}

Stochastic languages formalize probability distributions over words, with canonical representations via weighted automata~\cite{droste2007weighted} and probabilistic grammars~\cite{Denis2004}. These frameworks connect formal language theory with probabilistic modeling, supporting applications in grammar induction~\cite{Clark2010}, and natural language processing~\cite{Cohn2009}. Despite this extensive development, there is no established syntactic characterization that delineates exactly the class of stochastic languages definable by such models.

Early probabilistic automata, introduced by Rabin and Paz, localize randomness in the transition relation: outgoing transitions from each state form a stochastic distribution, and the probability of accepting a word is obtained by summing over all accepting runs labeled by that word~\cite{rabin1963probabilistic,paz1971probabilistic}. This yields a probabilistic acceptance semantics governed by cutpoint conditions. For $\omega$-languages, probabilistic Büchi automata assign Markov chains to runs and classify behaviors by almost-sure, positive, or threshold probability acceptance~\cite{baier2012probabilistic}. Chatterjee, Doyen, and Henzinger later generalized these models to probabilistic weighted automata, where transitions are probabilistic but evaluations are quantitative, defining random variables over runs rather than word-level distributions~\cite{chatterjee2009probabilistic,Bollig2012ProbabilisticKleene}. In concurrency theory, Segala and Lynch formalized probabilistic automata as systems with randomized transitions, introducing semantic notions of probabilistic bisimulation~\cite{segala1995probabilistic}. Across these variants, randomness arises operationally along runs rather than as a distribution over the word space.

On the other hand, Cost Register Automata (CRA), introduced by Alur et al.~\cite{Alur2013}, provide a deterministic, register-based model for quantitative functions in which a finite set of registers is updated along a run using algebraic operations. The model generalizes classical weighted automata, while retaining explicit operational semantics. More recently, Benalioua, Lhote, and Reynier analyzed CRA with linear and affine updates over fields, demonstrating equivalence in expressive power with weighted automata and developing minimization algorithms with matching complexity bounds~\cite{benalioua2024minimizing}. Together, these lines of work position CRA as a deterministic and compositional alternative to path-summing weighted models, offering a good fit for extending quantitative automata theory to probabilistic or normalized semantics.

A distinct line of work, initiated by Denis and Esposito~\cite{Denis2009LearningWA,Denis2004}, studied rational stochastic languages through multiplicity automata. While such models provide an algebraic vehicle for representing stochastic languages, the literature still lacks a clear syntactic or compositional characterization analogous to the Kleene–Schützenberger theorem for rational stochastic languages. In particular, it remains unclear whether stochastic languages admit a closed algebra of regular-like expressions corresponding precisely to the distributions generated by these automata. 

Our work addresses this gap by developing a cost register automata based framework for stochastic languages, proving a locally sub-stochastic characterisation, and introducing  Stochastic Regular Expressions that give an equivalence between algebraic and automata-theoretic descriptions of stochastic languages.

\section{Preliminaries and Notation}
\label{sec:notation}

In this section we recall the basic definitions and notation used throughout the paper.
Our focus is on \emph{quantitative languages} over non-negative reals,
their representation by \emph{weighted} and \emph{cost register automata}, 
and a few linear-algebraic tools used to reason about convergence.

\subsection{Quantitative Languages}
\label{sec:notation:language}

Let $\Sigma$ be a finite alphabet and $\Sigma^*$ the set of all finite words over $\Sigma$. For the purpose of our discussion, a \emph{quantitative language} over $\Sigma$ is a function $
f: \Sigma^* \to \R$, that assigns a real-valued weight to each word.
For a general treatment of quantitative languages, see \cite{droste2007weighted,Handbook-weighted-automata}.

\begin{definition}[Mass]
Let $L \subseteq \Sigma^*$ be a (possibly finite or infinite) set of words. The mass of a quantitative language $f: \Sigma^* \to \R$ on $L$ is defined as
\[
\mass{L}{f} = \sum_{w \in L} f(w),
\]
whenever this sum is well-defined. In particular, the total mass of $f$ is $\mass{\Sigma^*}{f}$.
\end{definition}

A quantitative language $f$ is {\emph finite-mass} if $
\mathrm{mass}_{\Sigma^*}$ is finite. \cite{Handbook-weighted-automata}

\begin{example}
\label{example-algebraic}
Let $\Sigma = \{a, b\}$. Let $f: \Sigma^* \to \R$ be defined as:
\[
f(w) = \frac{1}{(|w|+1)^2 \cdot 2^{|w|}},
\]
where $|w|$ denotes the length of $w$ (so that $|\varepsilon| = 0$). Then $f$ is a finite-mass language as $\mass{\Sigma^*}{f}= \frac{\pi^2}{6}$.
\end{example}

\begin{definition}[Stochastic Language] A quantitative language $f: \Sigma^* \to \R$ is stochastic if for all $w \in \Sigma^*$, $f(w) \geq 0$ and  $\mass{\Sigma^*}{f} = 1$. The set $\Stoch(\Sigma^*)$ denotes the set of all stochastic languages over $\Sigma^*$ \cite{Denis2004,Handbook-weighted-automata}.
\end{definition}

\begin{example}[Dirac Distribution]
\label{ex:dirac}
For a fixed string $w \in \Sigma^*$, the \emph{Dirac distribution} $\delta_w$ is a stochastic language over $\Sigma^*$ defined for all $u \in \Sigma^*$ as:
\[
\delta_w(u) = 
\begin{cases} 
1 & \text{if } u = w, \\
0 & \text{otherwise}.
\end{cases}
\]
\end{example}

\begin{example}
Consider the function $f$ in Example~\ref{example-algebraic}. The function $\overline{f}: w \mapsto \frac{6f(w)}{\pi^2}$ is a stochastic language. In general, if $f$ is a finite-mass language such that for all $w \in \Sigma^*$, $f(w) \geq 0$, and $\mass{\Sigma^*}{f} \neq 0$, the \emph{normalized} function 
$\overline{f}: w \mapsto f(w)/\mass{\Sigma^*}{f}$ is a stochastic language.
\end{example}

\begin{definition}
\label{def:operators}
We can now define some operations over stochastic languages. Let $f_1$ and $f_2$ be two stochastic languages, and let $\lambda \in (0, 1)$.
\begin{enumerate}
    \item Their convex combination parametrised by $\lambda$ is:
\[
\lambda f_1 + (1-\lambda)f_2: w \mapsto \lambda f_1(w) + (1-\lambda)f_2(w)
\]
\item Their Cauchy product is:
\[
f_1\cdot f_2: w \mapsto \sum_{w = uv}f_1(u) \cdot f_2(v)
\]
\item The {\emph discounted Kleene Star} of $f_1$ parameterised by $\lambda$ is: 
\[
(f_1)^*_\alpha: w \mapsto \sum_{k=1}^\infty \sum_{\substack{w_1,\ldots,w_k \in \Sigma^+ \\ w = w_1 \cdots w_k}} \alpha (1-\alpha)^{k-1} \prod_{i=1}^k f_1(w_i).
\]
\end{enumerate}
\end{definition}

Intuitively, \( f^*_\alpha(w) \) defines a distribution over strings obtained by concatenating \( k \) non-empty substrings, each independently drawn from \( r \), with the total number of substrings following a shifted geometric distribution with parameter \( \alpha \). This definition implicitly sets $r(\epsilon)=0$. This is in accordance with standard definitions \cite{bollig2015weighted}. It is immediate that stochastic languages are closed under the above three operators. 

\subsection{Quantitative Automata}

Weighted finite automata (WFAs) are the canonical model for quantitative languages
over semirings~\cite{bollig2015weighted,Handbook-weighted-automata,Balle2015LearningWA}. 
A {\em semiring} is a set $K$ with two binary operations $+$ and $\cdot$ and two constant elements $0$ and $1$ such that $(K, +, 0)$ is a commutative monoid, $(K,\cdot, 1)$ is a monoid, and $+$ distributes over $\cdot$ operation. We will work with the semirings $(\R, +, \cdot, 0, 1)$ and $(\Rp, +, \cdot, 0, 1)$.

\begin{definition}[Weighted Finite Automata (WFA)]
A weighted automaton $\mathcal{A}$ is a tuple $(\Sigma,Q,\lambda,M,\mu)$ over the semiring $K$ where:
\begin{enumerate}[label=\arabic*)]
\item $\Sigma$ is a finite alphabet,
  \item $Q$ is a finite state set,
  \item $\lambda:Q\to K$ is the initial weight vector,
  \item $M:\Sigma\to K^{Q\times Q}$ assigns a transition matrix $M_\sigma$ to each symbol,
  \item $\mu:Q\to K$ is the final weight vector.
\end{enumerate}
The semantics of $\mathcal{A}$ is the quantitative language
\[
\llbracket \mathcal{A} \rrbracket(w)
= \lambda^\top M_{w_1}\cdots M_{w_n}\mu
\quad \text{for } w=w_1\cdots w_n\in\Sigma^*.
\]
\end{definition}

The set $\Stoch^{rat}_K(\Sigma^*)$ represents the set of stochastic languages over $\Sigma^*$ that can be expressed by a weighted automaton over the semiring $K$. By restricting to the non-negative models in $\Srat$, we give up a bit of expressivity. Indeed, $\Srat \subsetneq \Stoch_{\R}^{rat}(\Sigma) \subsetneq \Stoch(\Sigma)$ \cite{Denis2004}.

Cost Register Automata (CRAs), introduced by Alur et al.~\cite{Alur2013}, offer
an alternative representation in which weights are stored explicitly in
registers and updated deterministically during the run.

\begin{definition}[Cost Register Automaton]
A real-valued CRA is a tuple
$\mathcal{A} = (\Sigma,Q,X,q_0,x_0,\delta,f,\mu)$ where
\begin{enumerate}[label=\arabic*)]
  \item $\Sigma$ is a finite alphabet, and $Q$ a finite state set;
  \item $X$ is a finite set of registers;
  \item $q_0\in Q$ is the initial state, and $x_0:X\to\mathbb{R}$ the initial valuation;
  \item $\delta:Q\times\Sigma\to Q$ is a deterministic transition function;
  \item $f:Q\times\Sigma\times\mathbb{R}^X\to\mathbb{R}^X$ specifies register updates;
  \item $\mu:Q\times\mathbb{R}^X\to\mathbb{R}$ is the output function.
\end{enumerate}
The semantics $\semantics{\mathcal{A}}:\Sigma^*\to\mathbb{R}$ is given by
\[
\semantics{\mathcal{A}}(w) = \mu(q_n,x_n),
\]
where $(q_0,x_0)\xrightarrow{w_1}(q_1,x_1)\xrightarrow{w_2}\cdots\xrightarrow{w_n}(q_n,x_n)$
is the unique run with $x_i=f(q_{i-1},w_i,x_{i-1})$.
\end{definition}

CRAs naturally capture deterministic weighted computations and are expressive enough to represent rational stochastic languages when instantiated with affine updates. We will use the cost register automata in Figure~\ref{fig:running-example} as our running example. 
\begin{figure}[t]
\centering
\begin{tikzpicture}[->,auto,accepting/.style=accepting by arrow]
    \node[state,initial,initial text = {$X:=1 \;\;Y:=0$}] (q1) {$q_1$};
    \node[state,
          accepting,%
          accepting text={$X$}] (q2)[right=5cm of q1] {$q_2$};
    \path (q1) edge [loop above]
                node [align=center] {$a \big/ X :=  Y$ \\ $a \big/ Y := \frac{1}{12}X+\frac{1}{4}Y$} (q1);
    \path (q1) edge [loop below]
                node [align=center] {$b \big/ X :=  0$ \\ $b \big/ Y := \frac{1}{12}X$} (q1);  
    \path (q2) edge [loop above]
                node [align=center] {$a \big/ X :=  \frac{1}{3}X$ \\ $a \big/ Y := \frac{1}{6} Y+ \frac{1}{3}$} (q2);
    \path (q2) edge [loop below]
                node [align=center] {$b \big/ X :=  \frac{1}{6}X$ \\ $b \big/ Y := \frac{1}{3}Y  + \frac{1}{6}$} (q2); 
    \path (q2)  edge [bend right=10]  
                node[above] {$c \Big/ X := 0\;\; Y:=\frac{1}{12}X$} (q1);
    \path (q1)  edge [bend right=10]
                node[below] {$c \Big/ X := X\;\; Y:=2Y + \frac{1}{2}$} (q2);                    
\end{tikzpicture}
\caption{A two-state cost register automaton (CRA) over alphabet $\{a,b,c\}$ with registers $X$ and $Y$. Transitions are labeled by input symbols and associated affine register updates. The initial valuation is $X=1, Y=0$, and the output in the accepting state $q_2$ is the value of register $X$.}
\label{fig:running-example}
\end{figure}
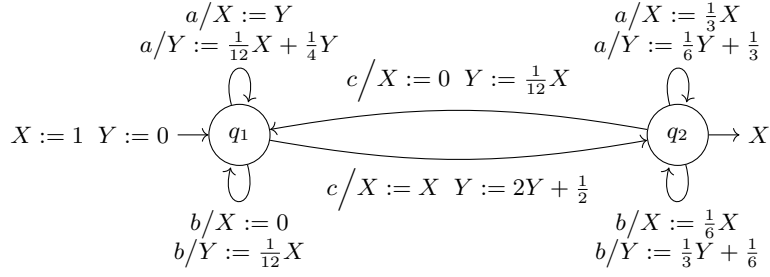
\subsection{Matrix Properties and Convergence}
\label{sec:notation:matrix}

We recall standard results from non-negative matrix theory
(see~\cite{HornJohnson2012}) that will be used to analyse convergence of weighted computations. The results discussed here are for real valued square matrices in $\R^{n\times n}$.

\begin{enumerate}
\item A non-negative square matrix $A$ is
\emph{irreducible} if for all $i,j$ there exists $k\ge1$ such that
$(A^k)_{ij}>0$.
\item The \emph{spectral radius} $\rho(\cdot)$ of a square matrix is the maximum of the absolute values of its eigenvalues. The spectral radius of a bounded operator is the supremum of the absolute values of the elements of its spectrum. 
\item A non-negative square matrix A is \emph{row sub-stochastic} if for every row $i$, $\sum_j A_{ij} \leq 1$. It is strictly sub-stochastic if for at least one row, the sum is stictly less than $1$.
\item (Perron-Frobenius Theorem) If a non negative matrix $A$ is irreducible, then:
\[
\min_i \sum_j A_{ij} \le \rho(A) \le \max_i \sum_j A_{ij}.
\]
In particular, if it is strictly sub-stochastic, then $\rho(A)<1$.
\item (Neumann Series) For a square matrix $A$, the series $I + A + A^2 + A^3 \ldots $ converges if and only if $\rho(A) < 1$. In this case, it converges to $(I - A)^{-1}$.
\end{enumerate}

$\rho(A)$ or $(I-A)^{-1}$ can be computed in $O(n^3)$ time.

\section{Problem Formulation and Decidability Results}
\label{sec:problem}

This work is motivated by a central semantic question: 
\begin{quote}
\emph{Given a cost register automaton (CRA), when does it define a stochastic language?}
\end{quote}

This problem connects the operational view of quantitative automata with the analytic notion of probability distributions over strings.

\begin{problem}[Stochasticity Decision Problem]
\label{prob:stochasticity}
Given a CRA $\mathcal{A}$, determine whether its semantics satisfies:
\begin{enumerate}[label=(\roman*)]
  \item Non-negativity: $\semantics{\mathcal{A}}(w)\ge0$ for all $w\in\Sigma^*$, and
  \item Normalization: the total mass equals one,
  \[
    \mass{\Sigma^*}{\semantics{\mathcal{A}}}
    = \sum_{w\in\Sigma^*}\semantics{\mathcal{A}}(w)
    = 1.
  \]
\end{enumerate}
\end{problem}

Intuitively, a CRA represents a stochastic language when its register updates behave like probability-preserving transitions and the total accumulated weight over all strings equals~1.  
In what follows, we chart the boundary between decidable and undecidable instances of this problem.

\subsection{Undecidability of the Stochasitcity Decision Problem}

We first consider the more basic question of whether a given CRA has \emph{finite total mass}.  
Even this relaxed problem turns out to be undecidable.

\begin{theorem}\label{thm:cra-undecidability}
Given a CRA $\mathcal{A}$ with affine register updates,
it is undecidable whether $\mass{\Sigma^*}{\semantics{\calA}}$ converges. 
\end{theorem}

\begin{proof}

We reduce from the \emph{Post Correspondence Problem (PCP)}.
Let the instance be given by pairs of strings $\{(f(i),g(i))\}_{i=1}^k$
over the digit alphabet $\Sigma=\{1,\dots,9\}$.
We construct a single-state CRA $\mathcal{A}$ (Figure~\ref{fig:cra-undecidability})
with registers $X_f,Y_f,X_g,Y_g$ initialized as $(0,1,0,1)$.
Upon reading symbol $\sigma\in\{1,\dots,k\}$, the updates are:
\[
\begin{aligned}
X_f &:= X_f + Y_f\cdot \mathrm{val}(f(\sigma)), &
Y_f &:= 10^{|f(\sigma)|}\cdot Y_f,\\
X_g &:= X_g + Y_g\cdot \mathrm{val}(g(\sigma)), &
Y_g &:= 10^{|g(\sigma)|}\cdot Y_g.
\end{aligned}
\]
The output $\mu(q,(X_f,Y_f,X_g,Y_g))=1$ iff $X_f=X_g$, and $0$ otherwise.

For a word $w=\sigma_1\cdots\sigma_n$, the registers $(X_f,Y_f)$ encode
the base-10 numeric value and positional multiplier for $f(w)$, and analogously for $g(w)$.
Hence, $\semantics{\mathcal{A}}(w)=1$ iff $f(w)=g(w)$.

If the PCP instance admits a solution $w$, then $\semantics{\mathcal{A}}(w^k)=1$ for all $k\ge1$, so
\[
  \mass{\Sigma^*}{\semantics{\mathcal{A}}}
  \ge \sum_{k=1}^\infty \semantics{\mathcal{A}}(w^k)
  = \infty.
\]
If no solution exists, $\semantics{\mathcal{A}}(w)=0$ for all nonempty $w$,
and $\semantics{\mathcal{A}}(\varepsilon)=1$, giving $\mass{\Sigma^*}{\semantics{\mathcal{A}}}=1$.
Thus, deciding convergence of $\mass{\Sigma^*}{\semantics{\mathcal{A}}}$ is equivalent to solving PCP. \qed
\end{proof}

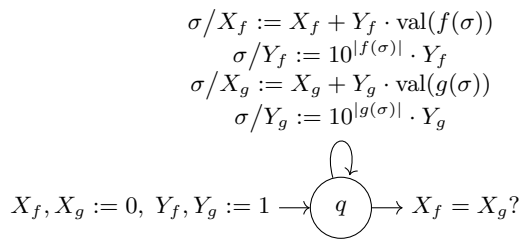
\begin{figure}[t]
\centering
\begin{tikzpicture}[->,auto,accepting/.style=accepting by arrow]
    \node[state,initial, accepting, 
          initial text={$X_f, X_g:=0,\; Y_f, Y_g:=1$},
          accepting text={$X_f = X_g?$}] 
          (q) {$q$};

    \path (q) edge [loop above] node[align=center] {
        $\sigma \big/ X_f := X_f + Y_f \cdot \mathrm{val}(f(\sigma))$ \\
        $\sigma \big/ Y_f := 10^{|f(\sigma)|} \cdot Y_f$ \\
        $\sigma \big/ X_g := X_g + Y_g \cdot \mathrm{val}(g(\sigma))$ \\
        $\sigma \big/ Y_g := 10^{|g(\sigma)|} \cdot Y_g$
    } (q);
\end{tikzpicture}
\caption{A single-state CRA used in the reduction from PCP in the proof of Theorem~\ref{thm:cra-undecidability}. The registers $X_f$ and $X_g$ accumulate the numerical value of interpretations of two transductions $f$ and $g$ over input $w$, and $Y_f$ and $Y_g$ maintain their lengths. The output function checks if $X_f = X_g$ and returns $1$, or returns $0$ otherwise.}
\label{fig:cra-undecidability}
\end{figure}

The proof of Theorem~\ref{thm:cra-undecidability} crucially relies on a non-polynomial output function (testing equality). This naturally leads us to consider the \emph{fully polynomial fragment}, where both the register updates and output functions are restricted to polynomials. We show a general property of polynomial dynamical systems:

\begin{theorem}
\label{thm:poly-cra-undecidability}
Let $p : \mathbb{R}^d \to \mathbb{R}^d$ be a polynomial map, $g : \mathbb{R}^d \to \mathbb{R}$ be a polynomial function, and $x_0 \in \mathbb{R}^d$ be an initial vector. It is undecidable to determine if the following series converges:

\begin{equation}
\label{eq:poly-series}
S(p, g, x_0) = \sum_{n=0}^{\infty} g(p^{(n)}(x_0)).
\end{equation}
\end{theorem}

\begin{proof}
We reduce from the halting problem for Turing machines. By classical constructions \cite{Bournez2007,Hainry2006}, for every Turing machine $M$ and input $w$, there exist $d$, a polynomial map $p:\mathbb{R}^d\to\mathbb{R}^d$,
and rational vectors $x_0,q\in\mathbb{Q}^d$ such that
the iteration $x_{n+1}=p(x_n)$ simulates the configuration of $M(w)$:
\begin{enumerate}
  \item if $M(w)$ halts after $N$ steps, then $x_n = q$ for all $n \ge N$, and
  \item if $M(w)$ does not halt, then $\|x_n - q\| \ge 1$ for all $n$.
\end{enumerate}

Define the function $g : \mathbb{R}^d \to \mathbb{R}$ by $
g(x) = \|x - q\|^2 = \sum_{i=1}^d (x_i - q_i)^2$. Then $g(x_n)=0$ eventually iff $M(w)$ halts. That is, 
\[
M(w) \text{ halts } \iff \sum_{n=0}^{\infty} g(p^{(n)}(x_0)) \text{ converges.}
\]
Thus, deciding convergence of Equation~\ref{eq:poly-series} is undecidable. \qed
\end{proof}

These results imply that checking whether a CRA has finite total mass is undecidable, both for CRAs restricted to \emph{affine} register updates with \emph{non-linear} outputs (Theorem~\ref{thm:cra-undecidability}), 
and for the \emph{fully polynomial} fragment with polynomial updates and outputs (Theorem~\ref{thm:poly-cra-undecidability}).

Moreover, given a CRA~$\mathcal{A}$, 
is it undecidable whether \(\semantics{\mathcal{A}}(w) \ge 0\) for all \(w \in \Sigma^*\). These follow from classical results for weighted finite automata (WFAs) where the boundedness problem over the non-negative rationals is known to be undecidable~\cite{LICS22}. 
Since CRAs strictly generalize WFAs~\cite{Alur2013}, 
it follows that neither non-negativity nor normalization can be algorithmically checked for general cost register automata.

\subsection{Linear Cost Register Automata}
\label{sec:linear-cra}

To recover decidability, we focus on a restricted and well-behaved fragment of cost register automata. 
First, to guarantee non-negativity of register valuations, we define the automaton over the positive semiring 
\((\mathbb{R}_{\ge 0}, +, \times, 0, 1)\) only, ensuring that all intermediate and output values remain non-negative. 
Second, we constrain both register updates and output expressions to be \emph{affine}, that is, linear transformations combined with additive constants. 
The analysis of affine CRAs can be further simplified using a standard linearization technique:

\begin{lemma}
\label{lemma:problem:linear}
The class of functions computable by CRAs with linear updates coincide with the class of functions computable by CRAs with affine updates.
\end{lemma}

\begin{proof}(Sketch)

Let $\mathcal{A} = (Q, \Sigma, X, \delta, \rho, q_0, \mu)$ be a CRA with affine updates of the form $x\mapsto A_{q,\sigma}x+b_{q,\sigma}$.
We add a new register $x_0$ that always holds $1$.
Replacing each affine update by the linear map
\(
(x,x_0)\mapsto
\begin{bmatrix}
A_{q,\sigma} & b_{q,\sigma}\\
0 & 1
\end{bmatrix}
(x,x_0)
\)
and output by $\mu'(x,x_0)=[c_q^\top\; d_q](x,x_0)$
preserves semantics. \qed
\end{proof}

These restrictions yield the class of \emph{linear CRAs}, which is placed within a well-studied computational model whose dynamics can be represented by systems of linear transformations over the non-negative reals. Linear CRAs admit an exact correspondence with weighted automata:

\begin{theorem}[Alur et al.~\cite{Alur2013}]
\label{thm:linear-wa-equivalence}
The class of functions computed by CRAs with linear updates
coincides with those computed by weighted automata.
\end{theorem}

The key idea is that a linear CRA with $k$ registers and $m$ states can be converted to an equivalent weighted automaton with $m \cdot k$ states by treating each register-value pair as a state, with transitions derived from the linear update functions. This equivalence immediately yields numerous closure properties and normal forms for linear CRAs. 
We can now prove our first main result:

\begin{theorem}
\label{thm:affine-decidable}
Given a CRA $\mathcal{A}$ over $\Rp$ with affine register updates and affine output function, the total mass $\mass{\Sigma^*}{\semantics{\calA}} = \sum_{w \in \Sigma^*} \semantics{\calA}(w)$ is computable.
\end{theorem}

\begin{proof}
By Lemma~\ref{lemma:problem:linear}, we may assume without loss of generality that $\calA$ has linear updates. By Theorem~\ref{thm:linear-wa-equivalence}, $\mathcal{A}$ is equivalent to a weighted automaton. Let $\mathcal{W}$ be its \emph{trimmed} version (containing only states that are both reachable from an initial state and co-reachable from a final state), with $n$ states $Q = \{1, \ldots, n\}$.

Let $M_\sigma \in \Rp^{n \times n}$ be the transition matrix for symbol $\sigma \in \Sigma$, where $(M_\sigma)_{i,j}$ is the weight of transitioning from state $i$ to state $j$ on $\sigma$. Let $\iota \in \Rp^n$ be the initial weight vector and $c \in \Rp^n$ the final weight vector.

Define the transition matrix $M = \sum_{\sigma \in \Sigma} M_\sigma$. The trimmed automaton $\mathcal{W}$ is strongly connected by construction, and hence $M$ a nonnegative irreducible matrix. The total weight of all strings of length $k$ is given by $c^\top M^k \iota$. Therefore, the total mass is:
\[
\mass{\Sigma^*}{\semantics{\calA}} = \sum_{w \in \Sigma^*} \semantics{\calA}(w) = {c}^\top \left( \sum_{k=0}^\infty M^k \right) {\iota}.
\]
This Neumann series converges if and only if the spectral radius $\rho(M) < 1$. The spectral radius of a nonnegative matrix irreducible matrix is computable as the largest real eigenvalue by the Perron-Frobenius theorem~\cite{HornJohnson2012}.
\begin{itemize}
    \item If $\rho(M) \geq 1$, then $\mass{\Sigma^*}{\semantics{\calA}}$ diverges.
    \item If $\rho(M) < 1$, then $\mass{\Sigma^*}{\semantics{\calA}}$ converges to ${c}^\top (I - M)^{-1} {\iota}$.
\end{itemize}
The computation involves computing $\rho(M)$ via eigenvalue computation in $O((|Q| \cdot |X|)^3)$ time, and if the spectral radius $\rho(B) < 1$, then solving the linear system $(I - M)x = \iota$ via Gaussian elimination in $O((|Q| \cdot |X|)^3)$ time. Thus, the total mass is computable in polynomial time in the size of the input CRA $\calA$. \qed
\end{proof}

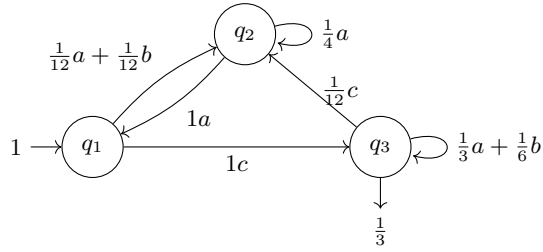
\begin{figure}[t]
\centering
\begin{tikzpicture}[->,auto,accepting/.style=accepting by arrow]
    \node[state,initial,initial text = $1$] (q1) {$q_1$};
    \node[state](q2)[right = 1.2cm of q1, yshift = 1.5cm]{$q_2$};
    \node[state,accepting,accepting text = $\frac{1}{3}$, accepting where = below] (q3)[right=3cm of q1]{$q_3$};
    \path (q1) edge [bend left = 12] node{$\frac{1}{12}a+\frac{1}{12}b$}(q2);
    \path (q2) edge[bend left = 12] node{$1a$} (q1);
    \path (q1) edge [below] node{$1c$} (q3);
    \path (q3) edge [right]node{$\frac{1}{12}c$} (q2);
    \path (q2) edge [loop right] node {$\frac{1}{4}a$}(q2);
    \path (q3) edge [loop right] node {$\frac{1}{3}a + \frac{1}{6}b$}(q3);
\end{tikzpicture}
\caption{A weighted automaton over $\Sigma = \{a,b,c\}$ with initial weight $1$ at $q_1$ and final weight $\tfrac{1}{3}$ at $q_3$. Transition labels denote linear combinations of weighted symbol transitions (e.g., $\tfrac{1}{12}a + \tfrac{1}{12}b$ means weight $\tfrac{1}{12}$ on $a$ and $\tfrac{1}{12}$ on $b$). This weighted automaton is equivalent to the cost register automata in Figure~\ref{fig:running-example}.}
\label{fig:running-example-wa}
\end{figure}

\begin{example}
The WFA in Figure~\ref{fig:running-example-wa} is equivalent to the CRA in Figure~\ref{fig:running-example}. For this WFA, we have:
\[
M_a+M_b+M_c =
\begin{pmatrix}
0 & \tfrac{1}{6} & 1\\
\tfrac{3}{4} & \tfrac{1}{4} & 0\\
0 & \tfrac{1}{12} & \tfrac{1}{2}
\end{pmatrix}.
\]
Its row sums are $(1\tfrac{1}{6},1,\tfrac{7}{12})$ and $\rho(M)=\tfrac{3}{4}$.
The total mass $c^\top (I-M)^{-1}\iota = 1$, confirming that the automaton
defines a normalized stochastic language.
\end{example}

In summary, while the stochasticity problem is undecidable in general (Theorem~\ref{thm:cra-undecidability}),
it becomes efficiently decidable for the affine fragment. For affine CRAs over $\Rp$: $\rho(M)<1$ ensures finite mass, and the positive semiring guarantees non-negativity.
Normalizing the output by $\mass{\Sigma^*}{\semantics{\mathcal{A}}}$ yields a stochastic language. 
This fragment is thus a useful model for representing probability distributions over strings. 


\section{Localizable Stochasticity and Sub-stochastic Cost}
\label{sec:sub-stochastic}

The decidability result in Theorem~\ref{thm:affine-decidable} establishes that the total mass 
of a linear CRA is computable whenever the combined transition 
matrix $M = \sum_{\sigma \in \Sigma} M_\sigma$ is 
irreducible and has spectral radius $\rho(M) < 1$. 

In contrast, classical probabilistic models enforce local properties: each 
state’s outgoing weights sum to~$1$, guaranteeing that the semantics forms 
a probability distribution without the need for a spectral analysis~\cite{bollig2015weighted}.

In this section, we show that every finite-mass weighted automaton
(and hence every CRA with affine updates defining a stochastic language)
admits an equivalent representation in which stochasticity is \emph{localizable}.

\subsection{Locally Sub-stochastic Weighted Automata}

Formally, we show that any WFA with only reachable and co-reachable states and $\rho(M) < 1$
can be normalized so that the outgoing weights from each state sum to at most~$1$.
We call this property \emph{local sub-stochasticity}.
Such normalization can be achieved by a similarity transform
based on the Perron–Frobenius eigenvector,
and it preserves the global semantics of the automaton.

\begin{proposition}
[Perron–Frobenius normalization]
\label{thm:diag-similarity}
Let $A \in \Rp^{n\times n}$ be an irreducible matrix with spectral radius $\rho(A) < 1$. 
Then there exists a positive diagonal matrix $D = \mathrm{diag}(d_1, \dots, d_n)$ with 
$d_i > 0$ such that $B = DAD^{-1}$ is strictly row sub-stochastic.    
\end{proposition}

\begin{proof}[Sketch]
By the Perron–Frobenius theorem, there exists a unique (up to scaling) positive eigenvector 
$v \in \Rp$ such that $Av = \rho(A) v$. Since $\rho(A) < 1$, we have $Av < v$ component-wise. 
Let $D = \mathrm{diag}(v_1, \ldots, v_n)$ and define $B = DAD^{-1}$. Then for each row~$i$,
\[
\sum_{j} B_{ij} 
= \sum_{j} \frac{v_i}{v_j} A_{ij}
= \frac{(Av)_i}{v_i}
\le 1,
\]
as required. As $\rho(A) < 1$, for some row $i$, $\sum_{j} B_{ij} < 1$. \qed
\end{proof}

Given a WFA $\mathcal{W} = (\Sigma, Q, \lambda, \{M_\sigma\}, \mu)$, define  
$M = \sum_{\sigma} M_\sigma$. Applying Theorem~\ref{thm:diag-similarity} to~$M$ yields a 
diagonal matrix~$D$ such that the similarity transform $M' = DMD^{-1}$ is row-substochastic. 
We then define a \emph{normalised automaton} $\mathcal{W}' = (\Sigma, Q, \lambda', \{M'_\sigma\}, \mu')$ with
\[
M'_\sigma = D M_\sigma D^{-1}, 
\qquad \lambda' = D^{-1}\lambda, 
\qquad \mu' = D\mu.
\]
The semantics of $\mathcal{W}'$ satisfies 
$\llbracket \mathcal{W}' \rrbracket(w) = \llbracket \mathcal{W} \rrbracket(w)$ for all $w \in \Sigma^*$, 
since each path weight is invariant under this rescaling.
Moreover, every row of $M'$ now sums to at most~$1$, hence the automaton is 
\emph{locally sub-stochastic}. 
Figure~\ref{fig:running-example-wa2} illustrates this normalization for the running example.

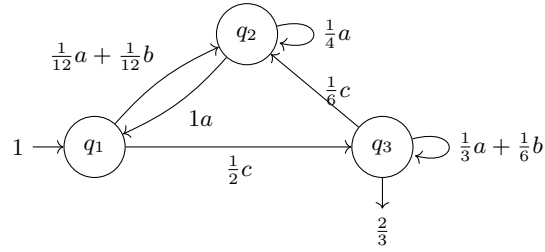
\begin{figure}[t]
\centering
\begin{tikzpicture}[->,auto,accepting/.style=accepting by arrow]
    \node[state,initial,initial text = $1$] (q1) {$q_1$};
    \node[state](q2)[right = 1.2cm of q1, yshift = 1.5cm]{$q_2$};
    \node[state,accepting,accepting text = $\frac{2}{3}$, accepting where = below] (q3)[right=3cm of q1]{$q_3$};
    \path (q1) edge [bend left = 12] node{$\frac{1}{12}a+\frac{1}{12}b$}(q2);
    \path (q2) edge[bend left = 12] node{$1a$} (q1);
    \path (q1) edge [below] node{$\frac{1}{2}c$} (q3);
    \path (q3) edge [right]node{$\frac{1}{6}c$} (q2);
    \path (q2) edge [loop right] node {$\frac{1}{4}a$}(q2);
    \path (q3) edge [loop right] node {$\frac{1}{3}a + \frac{1}{6}b$}(q3);
\end{tikzpicture}
\caption{A weighted automaton over $\Sigma = \{a,b,c\}$ with initial weight $1$ at $q_1$ and final weight $\tfrac{2}{3}$ at $q_3$. This automaton is semantically equivalent to the one in Figure~\ref{fig:running-example-wa}, and is sub-stochastic, that is the sum of outgoing weights of every state at most $1$.}
\label{fig:running-example-wa2}
\end{figure}

This gives us characterisation of rational stochastic languages:
\begin{theorem}[Local Sub-stochastic Characterisation]
\label{thm:local-substochastic}
A quantitative language \( f \) belongs to the class of stochastic rational languages
\(\Srat\) if and only if there exists a \emph{locally sub-stochastic} weighted automaton
\(\mathcal{W}\) such that $\semantics{\mathcal{W}} = f$.
\end{theorem}

The notion of rational stochastic languages has been investigated in several equivalent frameworks,
including probabilistic automata and rational series over the non-negative reals
(\cite{bollig2015weighted,Denis2004}).
The local sub-stochastic normalization bridges these perspectives,
revealing that the analytic convergence condition $\rho(M) < 1$
is equivalent to the existence of a locally normalizable form.

\subsection{Stochastic Regular Expressions}
\label{subsec:sre}

This local sub-stochasticity enables a Kleene-Schützenberger 
characterization of rational stochastic languages. 
We introduce \emph{stochastic regular expressions} (SREs), which extend ordinary regular 
expressions with
a probabilistic choice and discounted iteration operators. 
SREs capture exactly the class of quantitative languages definable by a sub-stochastic weighted automata.

\begin{definition}[Syntax of Stochastic Regular Expressions]
\label{def:sre}
Let $\Sigma$ be a finite alphabet. The grammar of SREs is
\[
r ::= \delta_\sigma 
\mid \alpha r_1 + (1-\alpha) r_2 
\mid r_1 \cdot r_2 
\mid r^*_\alpha
\]
where $\sigma \in \Sigma$, $r_1,r_2$ are SREs, and $\alpha \in (0,1)$.
\end{definition}

Here $\delta_\sigma$ denotes the Dirac distribution concentrated on~$\sigma$ (see Example~\ref{ex:dirac}), the operator 
$\alpha r_1 + (1-\alpha)r_2$ represents convex combination, $r_1\cdot r_2$ denotes the 
Cauchy product (probabilistic concatenation), and $r^*_\alpha$ denotes a \emph{discounted Kleene star}, where the number of repetitions follows a 
shifted geometric distribution with parameter~$\alpha$ (see Definition~\ref{def:operators}). We can now prove our characterisation theorem:

\begin{theorem}[Kleene-Schützenberger Characterisation of Stochastic Regular Languages]
\label{thm:kleene-characterisation}
The class of stochastic regular languages, denoted by~$\Srat$, is the smallest class of 
quantitative languages over~$\Sigma$ that contains all Dirac distributions 
$\{\delta_\sigma \mid \sigma \in \Sigma\}$ and is closed under convex combinations, Cauchy products, and discounted Kleene star.
\end{theorem}

The proof follows by constructing equivalences between locally sub-stochastic weighted automata and stochastic regular expressions (SREs). It parallels the construction between nondeterministic automata and regular expressions.

\begin{lemma}\label{lemma: automata-to-SRE}
Every language defined by a 
locally sub-stochastic weighted automaton can be represented by a Stochastic Regular 
Expression (SRE).
\end{lemma}

\begin{proof}[Sketch]
We adapt the standard state-elimination construction, exploiting that local sub-stochasticity ensures all geometric 
series that appear during elimination converge.

Given a normalized weighted automaton 
$\mathcal{W} = (\Sigma,Q,\lambda,T,\mu)$, augment it with a new start state $q_s$ and final 
state $q_f$, connected by weighted $\varepsilon$-transitions according to $\lambda$ and~$\mu$. 
Then iteratively eliminate all intermediate states $q_k \in Q$, replacing paths 
$q_i \to q_j$ that pass through~$q_k$ by a single edge labelled with the expression
\[
R_{ij} = T(q_i,\sigma,q_k)\, \delta_\sigma \cdot 
\Bigl(\sum_{\sigma} \tfrac{T(q_k,\sigma,q_k)}{W_k}\delta_\sigma\Bigr)^*_{1-T(q_k,\Sigma,q_k)}
\cdot R_{kj},
\]
where $W_k$ is the total outgoing weight from~$q_k$.  
Because $T(q_k,\Sigma,q_k)<1$, every such geometric expansion is well defined.  
After all eliminations, the remaining edge $q_s \to q_f$ is labeled with an SRE~$r$ 
satisfying $\llbracket r \rrbracket = \llbracket \mathcal{W} \rrbracket$. \qed
\end{proof}

Note that the size of the SRE can be exponential in the size of the weighted automaton~\cite{Handbook-weighted-automata}. Conversely:

\begin{lemma}
\label{lemma: SRE-to-automata}
Every SRE $r$ defines a stochastic language that can be realized by a 
locally \emph{stochastic weighted automaton}, that is, for every state, the sum of weights of outgoing transitions is exactly $1$. \qed
\end{lemma}

\begin{proof}
The proof proceeds by structural induction on $r$.
For $r = \delta_\sigma$, define the two-state automaton 
$\calA_\sigma = (\Sigma,\{q_i,q_f\},\lambda,T,\mu)$ with 
$\lambda(q_i)=1$, $T(q_i,\sigma,q_f)=1$, and $\mu(q_f)=1$. 
Clearly $\sum_w \llbracket \calA_\sigma \rrbracket(w)=1$. For the operators:

\begin{enumerate}
    \item (Convex Combination) For $r = \alpha r_1 + (1-\alpha)r_2$, given automata $\calA_1,\calA_2$ for $r_1,r_2$ on disjoint 
state sets, scale the initial weights of $\calA_1$ and $\calA_2$ by $\alpha$ and $(1-\alpha)$, 
respectively. Linearity of semantics yields 
$\sum_w \llbracket \calA \rrbracket(w)=1$.
\item (Cauchy Product) For $r = r_1 \cdot r_2$, connect every final state of $\calA_1$ to every initial state of 
$\calA_2$ by $\varepsilon$-transitions weighted by $\mu_1(p)\lambda_2(q)$.  
This realizes the Cauchy product of distributions and preserves total mass.
\item (Discounted Kleene Star) 
For $r = (r_1)^*_\alpha$, introduce 
discount factor $\alpha \in (0,1)$.  
Let $\calA_1$ realize $r_1$.  
Each repetition contributes a multiplicative $(1-\alpha)$ factor, and termination 
contributes $\alpha$. Thus
\[
\sum_{w\in \Sigma^+} \llbracket \calA \rrbracket(w)
= \sum_{k\ge1} \alpha (1-\alpha)^{k-1}
  \Bigl(\sum_w \llbracket \calA_1 \rrbracket(w)\Bigr)^k = 1.
\]
\end{enumerate}

By construction, the total mass of the automaton is exactly $1$ and for every state, the sum of weights of outgoing transitions is exactly $1$. \qed
\end{proof}

Together, these two lemmas complete the proof for Theorem \ref{thm:kleene-characterisation}. Intuitively, the nested discounted stars correspond to loops in the automaton whose expected number of traversals follows a geometric law. For our running example, the stochastic regular expression is: 

\begin{equation*}
\left(\left(\frac{1}{4}a + \frac{1}{4} b + \frac{1}{4}c \left(\frac{2}{3}a + \frac{1}{3}b\right)^*_{\frac{1}{2}}c\right)^*_{\frac{2}{3}} (a)^*_{\frac{3}{4}}a\right)c\left(\frac{2}{3}a + \frac{1}{3}b\right)^*_{\frac{1}{2}}
\end{equation*}

In summary, Theorems~\ref{thm:local-substochastic} and \ref{thm:kleene-characterisation} provide a local characterisation of probabilistic behaviour while preserving global convergence semantics. 

\section{Approximation and Statistical Inference for Stochastic Languages}
\label{sec:approx-learn}

In this section, we move beyond exact
characterization to the \emph{quantitative} analysis of stochastic languages:
their approximation, sampling, and learnability. We also introduce a minimal
testing framework that connects these ideas to classical distribution testing
in statistics.

\paragraph{Universal Approximation.}
We begin by showing that stochastic regular languages ($\Srat$) are dense
in the space of all stochastic languages ($\Stoch(\Sigma)$) under the $\ell_1$ (total variation) metric. This follows
from the simple observation that any distribution over $\Sigma^+$ can be
approximated arbitrarily well by truncating its support.

\begin{theorem}[Universal Approximation]
\label{thm:universal-sup-approx}
For every $f \in \Stoch(\Sigma^*)$ and $\varepsilon > 0$, there exists
$r \in \Srat$ such that $\| f - r \|_1 < \varepsilon$.
\end{theorem}

\begin{proof}
Since $f$ is a probability distribution over the countable set $\Sigma^*$,
choose a finite $S \subseteq \Sigma^+$ such that $\sum_{w \in S} f(w) >
1-\varepsilon/2$. Define $r_S$ as the normalized restriction of $f$ to $S$:
\[
r_S(w) =
\begin{cases}
\dfrac{f(w)}{f(S)} & \text{if } w \in S,\\[4pt]
0 & \text{otherwise.}
\end{cases}
\]
Then $r_S$ is a stochastic language with finite support and therefore
$r_S \in \Srat$. Moreover,
$\|f - r_S\|_1 = 2(1 - f(S)) < \varepsilon$. \qed
\end{proof}

Hence $\Srat$ is dense in $\Stoch(\Sigma)$ under $\ell_1$ (and therefore also
under $\ell_\infty$). Intuitively, the convex hull of Dirac distributions
$\{\delta_w \mid w \in \Sigma^+\}$ suffices to approximate any stochastic
language arbitrarily closely. This motivates a quantitative analogue of
classical automata minimization.

\begin{problem}[Minimal-State $\varepsilon$-Approximation]
\label{prob:min-k}
Given $r \in \Stoch(\Sigma^*)$ and $\varepsilon > 0$, determine the minimal
$k \in \mathbb{N}$ for which there exists a $k$-state weighted automaton
$\mathcal{A}$ satisfying
$\| r - \llbracket \mathcal{A} \rrbracket \|_1 \le \varepsilon$.
\end{problem}

When $\varepsilon = 0$ and $r$ is rational, the minimal $k$ coincides with the
Hankel rank of $r$, computable via standard spectral methods
(e.g.~\cite{Carlyle1971Weighted,Schutzenberger1961Hankel,Denis2009LearningWA}).
For $\varepsilon > 0$, Problem~\ref{prob:min-k} becomes that of computing an
\emph{approximate Hankel rank}, a continuous relaxation of minimality that
remains open in general. 

\paragraph{Sampling from Stochastic Regular Expressions.}
Stochastic regular expressions (SREs) admit exact generative sampling
procedures aligned with their semantics. Sampling proceeds recursively as
follows:
\begin{enumerate}[label=(\alph*), leftmargin=*]
  \item Atomic symbol: For $r = \delta_\sigma$, output $\sigma$.
  \item Convex combination: For
        $r = \alpha r_1 + (1-\alpha)r_2$, sample
        $b \sim \mathrm{Bernoulli}(\alpha)$ and recurse into $r_b$.
  \item Concatenation: For $r = r_1 \cdot r_2$, sample
        $w_1 \sim r_1$ and $w_2 \sim r_2$, and output $w_1 w_2$.
  \item Discounted star: For $r = r^*_\alpha$, draw
        $k \sim \mathrm{Geom}(\alpha)$ and concatenate
        $w_1, \ldots, w_k$ sampled independently from $r$.
\end{enumerate}
The expected length of a sample from $r^*_\alpha$ is $(1-\alpha)/\alpha$, hence
finite for $\alpha > 0$. Each operator preserves independence and normalization,
ensuring that the sampling process exactly realizes
$\llbracket r \rrbracket$.

\paragraph{Learnability.}
The density result of Theorem~\ref{thm:universal-sup-approx} implies that
stochastic languages can be approximated by finite-parameter models, suggesting
learnability from data. Two settings naturally arise: 

\begin{enumerate}
    \item Given i.i.d.\ samples from an unknown $r \in \Stoch(\Sigma^*)$, one can fit a weighted automaton by empirical estimation of Hankel moments, following
the spectral learning paradigm for rational series
\cite{Balle2015LearningWA,Denis2009LearningWA}.
Sample complexity bounds follow from standard concentration inequalities for
low-rank moment matrices.
\item If an oracle provides evaluation access to $\llbracket r \rrbracket(w)$ for
queries $w \in \Sigma^*$, the parameters of a linear CRA or weighted automaton
can be recovered by solving a finite system of linear equations~\cite{Angluin1987LStar}. Under regularity assumptions, this system is solvable in
polynomial time.
\end{enumerate}

\paragraph{Distribution Testing.}
Finally, we consider the problem of distribution testing:

\begin{problem}[Distribution Testing]
\label{prob:identity}
Given $Q \in \Srat$ and sample access to $P \in \Stoch(\Sigma)$, design a
randomized algorithm that, for tolerance $\varepsilon > 0$ and confidence
$\delta \in (0,1)$, satisfies:
\[
\begin{split}
\text{Accept}, \quad & \|P-Q\|_1 < \varepsilon,\\
\text{Reject},\quad & \|P-Q\|_1 > \varepsilon,
\end{split}
\]
with probability at least $1-\delta$.
\end{problem}

This extends the classical problem on finite domains \cite{Batu2013Testing,Canonne2022Price}. A naive solution proceeds by a transforming it to the finite-support case by
truncating the distribution tail. Let $\theta \in \mathbb{N}$ be the smallest
length such that $\sum_{w \in \Sigma^{\le \theta}} Q(w) > 1-\varepsilon/3$.
Construct the restricted distributions $P_{\le\theta}$ and $Q_{\le\theta}$ on
$\Sigma^{\le \theta}$ and apply any finite-domain tolerant tester
(e.g.~\cite{Canonne2022Price}) with thresholds $(\varepsilon/3, \varepsilon)$.
This yields correctness guarantees identical to the finite setting, with
sample complexity $\widetilde{\Theta}(|\Sigma|^{\frac{\theta+1}{2}}/\varepsilon^2)$.

\section{Conclusion and Future Work}

In this work, we advance the theory of stochastic languages through the lens of sub-stochastic weighted automata, establishing fundamental computational boundaries for quantitative automata models. Our main contributions are:

\begin{enumerate}
    \item We show that determining whether a cost register automaton defines a stochastic language is undecidable in general, even for affine register updates (Theorem~\ref{thm:cra-undecidability}) and the fully polynomial fragment (Theorem~\ref{thm:poly-cra-undecidability}). This holds for finite mass, non-negativity, and exact normalization.
    
    \item By restricting to cost register automata over non-negative reals with affine updates, we recover decidability and efficient computation. For this fragment, expressively equivalent to weighted finite automata, the total mass is computable in polynomial time via spectral analysis (Theorem~\ref{thm:affine-decidable}).
    
    \item We establish that every stochastic weighted automaton admits a semantically equivalent \emph{sub-stochastic} representation via \emph{Perron-Frobenius normalization} (Theorem~\ref{thm:local-substochastic}). This local characterization mirrors classical probabilistic automata, enabling convergence verification syntactically, without global spectral analysis.
    
    \item We introduce \emph{stochastic regular expressions} (SREs) as a compositional, probabilistic analogue of regular expressions and provide a Kleene-Schützenberger characterization for rational stochastic languages (Theorem~\ref{thm:kleene-characterisation}).
\end{enumerate}

\subsection*{Open Problems and Future Directions}

Despite these advances, several fundamental questions remain open. We highlight three directions where sub-stochastic representations could play a decisive role:

\begin{enumerate}
    \item \emph{Decidability Beyond Linear Updates. (Problem \ref{prob:stochasticity})} Can syntactic conditions on non-linear cost register automata ensure stochasticity decidability? One can look at functions $f$ that satisfy $\|f(x)-f(y)\|\le L\|x-y\|$ with $L<1$. Banach’s fixed-point theorem guarantees convergence to a unique fixed point. For CRAs with contracting updates, we can explore if stochasticity is decidable and if spectral methods extend via linearization around fixed points? More generally, polynomial dynamical systems with well-understood asymptotics may admit decidable fragments. Investigating \emph{semantic classes} of well-behaved CRAs could significantly expand the applicability of sub-stochastic methods. Approximate notions of stochasticity, e.g., $\varepsilon$-closeness, may suffice for practical purposes and connect to numerical optimization.

    \item \emph{Complexity of Minimal-State $\varepsilon$-Approximation (Problem~\ref{prob:min-k}).} Determining the minimal number of states to $\varepsilon$-approximate a stochastic language is a central open question. This connects to low-rank matrix approximation and non-negative factorization. Approximate minimal-state representations determine sample complexity in learning stochastic languages and illuminate the intrinsic complexity of probabilistic sequential models.
    \item \emph{Distribution Testing (Problem~\ref{prob:identity}).} Efficient, instance-optimal testing algorithms are needed for stochastic languages, particularly for validating large language models. Some important questions are: Can we exploit local sub-stochastic properties of distributions for efficient testing? Can we use these methods to perform efficient black-box testing for language models? This problem bridges theory and practice, providing foundations for \emph{trustworthy AI} and enabling rigorous, sample-efficient validation of probabilistic generative models.
\end{enumerate}

In conclusion, our work establishes the connection between stochastic languages and their sub-stochastic representations. As probabilistic models become central to computing, formal methods to specify, analyse, and verify their distributional behaviour are increasingly urgent. Viewing stochastic languages through sub-stochastic models offers a principled foundation, laying the groundwork for a formal \emph{theory of probabilistic computation}.


\clearpage

\appendix

\end{document}